\documentclass[a4paper,11pt]{article}

\usepackage{fullpage}

\usepackage[utf8x]{inputenc}
\usepackage{authblk}
\usepackage{cite}

\usepackage[svgnames]{xcolor} 

\usepackage{times} 

\usepackage{graphicx} 

\usepackage{booktabs} 
\usepackage[font=small,labelfont=bf]{caption}

\usepackage{amsfonts, amsmath, amsthm, amssymb}
\usepackage{wrapfig}
\usepackage{bbm,bm}
\usepackage{MnSymbol}
\usepackage{enumerate}
\usepackage{enumitem}

\usepackage{multicol}
\usepackage{float}

\usepackage[left=1.5cm,right=1.5cm,top=1.5cm,bottom=1.0cm,includeheadfoot]{geometry}

\usepackage{tikz}
\usetikzlibrary{decorations.pathmorphing,fit,decorations.pathreplacing} 
\usetikzlibrary{shapes, matrix} 
\usetikzlibrary{arrows} 
\usetikzlibrary{calc} 
\usetikzlibrary{positioning}

\tikzstyle{block} = [draw,rectangle,thick,minimum height=2em,minimum width=2em]
\tikzstyle{sum} = [draw,circle,inner sep=0mm,minimum size=2mm]
\tikzstyle{connector} = [->,thick]
\tikzstyle{line} = [thick]
\tikzstyle{branch} = [circle,inner sep=0pt,minimum size=1mm,fill=black,draw=black]
\tikzstyle{guide} = []
\tikzstyle{snakeline} = [connector, decorate, decoration={pre length=0.2cm, post length=0.2cm, snake, amplitude=.4mm, segment length=2mm},thick, magenta, ->]

\numberwithin{equation}{section}

\def\begeq{\begin{equation}}
\def\endeq{\end{equation}}

\newtheorem{lemma}{Lemma}[section]

\newtheorem{theorem}{Theorem}[section]
\newtheorem{proposition}{Proposition}[section]
\newtheorem{definition}{Definition}[section]

\def\begp{\begin{proposition}}
\def\endp{\end{proposition}}
\def\begl{\begin{lemma}}
\def\endl{\end{lemma}}
\def\begt{\begin{theorem}}
\def\endt{\end{theorem}}
\def\begd{\begin{definition}}
\def\endd{\end{definition}}

\newcommand{\la}{\langle}
\newcommand{\ra}{\rangle}

\newcommand{\tr}{\text{tr}}

\newcommand{\eps}{\varepsilon}

\newcommand{\SP}{\text{ }}

\newcommand{\nl}{\newline}
\newcommand{\ONE}{\mathbbm{1}}
\newcommand{\RR}{\mathbbm{R}}
\newcommand{\CC}{\mathbbm{C}}

\newcommand{\NN}{\mathbbm{N}}

\title{Completely Device Independent Quantum Key Distribution}
\author[1,3]{Edgar A Aguilar}
\author[2,3]{Ravishankar Ramanathan}
\author[4]{Johannes Kofler}
\author[2,3]{Marcin Paw\l owski}
\affil[1]{\textit{Institute of Mathematics, University of Gdansk,
80-952 Gdansk, Poland}} \affil[2]{\textit{Institute of Theoretical
Physics and Astrophysics, University of Gdansk, 80-952 Gdansk,
Poland}} \affil[3]{\textit{National Quantum Information Center of
Gdansk, 81-824 Sopot, Poland}}
\affil[4]{\textit{Max-Planck-Institute of Quantum Optics, 85748
Garching, Germany}}
\date{}
\begin{document}

\maketitle \vspace{-1cm}
\begin{abstract}
Quantum key distribution (QKD) is a provably secure way for two
distant parties to establish a common secret key, which then can be
used in a classical cryptographic scheme. Using quantum
entanglement, one can reduce the necessary assumptions that the
parties have to make about their devices, giving rise to
device-independent QKD (DIQKD). However, in all existing protocols
to date the parties need to have an initial (at least partially)
random seed as a resource. In this work, we show that this
requirement can be dropped. Using recent advances in the fields of
randomness amplification and randomness expansion, we demonstrate that
it is sufficient for the message the parties want to communicate to
be (partially) unknown to the adversaries -- an assumption without
which any type of cryptography would be pointless to begin with. One
party can use her secret message to locally generate a secret
sequence of bits, which can then be openly used by herself and the
other party in a DIQKD protocol. Hence, our work reduces the
requirements needed to perform secure DIQKD and establish safe
communication.
\end{abstract}

\section{Introduction}
Within the advancing quantum information technologies, quantum key
distribution (QKD) is arguably the technologically most advanced
field and has already entered the market with working product
solutions. In this quantum cryptographic protocol, two parties
usually named Alice and Bob exploit the laws of quantum physics to
produce a shared random key that remains unknown to the rest of the
world and which can then be used as a one-time pad in a classical
cryptographic scheme \cite{BB84,Eke91,Rev1,Rev2}.

The security of entanglement-based QKD protocols relies on the
violation of a Bell inequality \cite{Bell64} using pairs of quantum
entangled particles shared by Alice and Bob. Remarkably, it has been
shown that such entanglement-based protocols allow
device-independent QKD (DIQKD) \cite{BHK,VV}, in which the two
parties need not make any assumptions about the inner workings of
their devices, in particular the source that produces the systems
which the parties measure as well as their own measurement devices.
In principle, the measurement apparatuses can be bought from an
untrusted party, the eavesdropper Eve, and the particle pair source
can even be operated by her (as long as, for example, there are no hidden transmitters in the devices). Alice and Bob can still extract a
secret key by sufficiently violating a Bell inequality. However,
they are required to have access to a certain amount of randomness
which they use for their setting choices \cite{KPB,BPPW}. This is
related to the fact that no Bell inequality can be derived without
the "freedom-of-choice assumption" \cite{BCHS}.

For long messages, Alice and Bob need many settings to produce a key
long enough, such that it becomes infeasible to invent their own
random sequences bit by bit out of their heads. Hence, their
settings need to be produced by some sort of fast device. Such a
random number generator and its corresponding randomness must be
considered a resource in the protocol. However, it is impossible to
verify that any given random number generator is not determined by
some underlying mechanism which is simply unknown to the user but
not to the eavesdropper. Clearly, Alice and Bob should not buy their
randomness generators from Eve. Therefore, in some sense, the
assumptions in DIQKD are contradictory. While one does not trust the
measurement devices, one trusts the random number generators used
for the setting choices. Recent developments (e.g.\ regarding the
Dual Elliptic Curve Deterministic Random Bit Generator) have shown
that this trust can be problematic \cite{NSA}.

It is indeed possible to reduce the amount of required initial
randomness via randomness amplification and expansion. These
protocols exploit quantum correlations also in a device-independent
way \cite{CR,PAMetal,Galetal,MS,CSW,BRGetal}. The former field studies how,
given a source of imperfect randomness which is partially correlated
to the external world, one can produce a short string which is
completely uncorrelated and safe. The latter studies how, given a
finite amount of perfect random bits one can produce a longer
(potentially unbounded) random bit string. Both of these processes
have been generalized recently to achieve unbounded random strings
from finite min-entropy sources \cite{MS,CSW}. However, both
protocols require an initial (at least partially) random seed, and
there is no apparent way of getting around this assumption if one
wants to stick to the device-independent scenario.

We define a completely DIQKD (CDIQKD) protocol to be one which is
not only device-independent regarding the measurement apparatuses
and the pair source but which also does not need to make any
assumptions about the setting generators or initial random seeds. It
seems that this is an impossible task. The QKD community has been
working within the paradigm that if not at least one of the parties does
not have an initial (at least partially) random source, then sending
safe messages is not feasible.

In this paper, however, we show that the obstacle is surmountable
and that CDIQKD is indeed possible. The solution lies in the
observation that Alice and Bob do not really need their settings to
be random with respect to the whole universe. They only need
randomness with respect to Eve. Therefore, having a string which is
random to Eve and the devices used in the protocol is sufficient,
even though the string is not random with respect to an honest party
like Alice. And there is one thing, which is random to Eve due to
the fundamental underlying assumption in cryptography: the message
$\mathcal{X}$ which Alice wants to send to Bob. Without this trivial
assumption -- so basic that it usually is not even mentioned --,
there is no reasonable cryptographic task in the first place.

Our procedure seems counter-intuitive and risky, but in this paper
we give a proof of principle that it is secure. In the following, we
will show that Alice can use her secret message to locally generate
a secret sequence of bits, which can then be used by herself and Bob
as the settings in an entanglement-based QKD protocol.

\section{Background and Assumptions}
We will work with the standard QKD assumptions which, for the sake
of clarity, are listed below.

\vspace{0.5cm}\noindent\textit{Quantum Key Distribution Assumptions:}

\begin{itemize}[leftmargin=0.5cm]
\item[1.] \textit{Shielding.} A no-signaling condition is imposed on the components of each device, as well as between devices in both parties' laboratories.
\item[2.] \textit{Authenticated classical communication channel between parties.} This is not assumed to be secure, i.e.\ any classical communication is accessible to Eve. Furthermore, we consider this authenticated channel to be available to the parties as a black box resource, that was for instance previously established using a secret key.
\item[3.] \textit{Restriction to quantum theory.} The adversary can only prepare devices following the laws of quantum mechanics. In particular, she does not possess arbitrary no-signaling devices.
\item[4.] \textit{Message with randomness.} Alice possess a message $\mathcal{X}$ with $k$ min-entropy with respect to Eve and the devices, and can estimate this value.  $k$ needs to be sufficiently large.
\end{itemize}

These are the fundamental assumptions, without any of which the protocol could not guarantee security. For example, without assumption 1, there could be a transmitter in the devices telling Eve everything that is going on in the laboratories (including the secret message), or Eve could manipulate the devices externally. Furthermore, the protocols work assuming a Bell inequality was violated for which the components of the physical devices must not communicate, which for example, could be enforced by a spacelike separation. Assumption 2 is needed to avoid the ``Man in the Middle'' attack, even though this classical channel is accessible to the adversary. In the present work we consider the channel as a black box resource, see the Discussion for an elaboration. Assumption 3 may seem restrictive at a mathematically fundamental level, but this is also a standard assumption for security proofs such as in \cite{VV,CSW,MS,DPVR}, since super-quantum correlations have not been observed experimentally. Finally, our main assumption is that Alice's message $\mathcal{X}$ has some conditional min-entropy with respect to Eve and the devices, and that Alice is able to estimate this value. We argue that this is a sound assumption (and indeed usually left implicit), since if the message was not at least partially random to Eve, then performing a QKD protocol would lose all its point to begin with, as was already suggested in the concluding remark of \cite{ER}.

In this article, we think of \textit{conditional min-entropy}
$H_{\text{min}}$ operationally. If we have the classical quantum
state $\rho_{XE} = \sum_x P_X(x) |x\ra\la x| \otimes \rho_E^x$,
classical over $X$ and quantum over $E$, then the probability
that party $E$ correctly guesses the value of the random variable
$X$ is:
\[
p_{\text{guess}}(X|E) = \sum_x P_X(x)\tr[F_x \rho_E^x] = 2^{- H_{\text{min}}(X|E)_\rho}
\]
where $\{F_x\}$ is the optimal POVM on $E$ \cite{KRS}. In words, this means that the min-entropy quantifies how much of the string $X$ is unknown to system $E$. This is the standard way of quantifying randomness, by which we mean how much of a variable is unpredictable to a third party. In that sense, the ``most random variable" $X$ corresponds to the uniform distribution $U_X$ which is completely independent from everything else. In that case the min-entropy is simply the number of random bits, $H_{\text{min}}(X|E) = |X|$.


By \textit{Randomness Extractors} Ext($k,\eps$), we refer to
deterministic algorithms, which take a source $X$ with min-entropy
$k$, together with a uniform random seed of length $d$, to produce
an output of length $m$, which is an $\eps$-distance from the
uniform distribution. We shall use Trevisan's extractor \cite{Tre},
which was proven to be secure against quantum adversaries in
\cite{DPVR}, following the works of \cite{RRV,TSSR}. See Appendix B,
for a rigorous treatment.

A powerful observation which we will also need is the
\textit{Equivalence Lemma} from \cite{CSW}. The lemma states that
the security of  protocols using perfectly random strings depends on
these strings being perfectly random to the devices, and requiring
perfect randomness to both the devices and the adversaries is not
necessary. This is formally stated in the appendix as Lemma A.1.
Since we are assuming that Eve doesn't signal to the devices, the
important thing then is that the devices are not preprogrammed to
receive certain inputs. If during the protocol Eve learns more about
what random seeds Alice and Bob will use, then even if she adapts her eavesdropping strategy she cannot gain any advantage, so long as the devices were distributed beforehand.

Chung, Shi, and Wu devised a protocol which can amplify any finite
source with min-entropy $k$, by using Trevisan's Extractors
Ext($k,\eps$) \cite{CSW}. They coined this procedure
\textit{Physical Randomness Extraction}, because they rely on
physical procedures which extract randomness in a secure manner
through Bell tests.  Their solution is to use Ext($k,\eps$) on the
min-entropy source with all $2^d$ possible seed strings of length
$d$, and feeding each hashed output to different implementations of the
physical extraction protocol (which here will be a randomness expansion protocol). By different implementations, we mean
using new devices on each run of the physical protocol as to
guarantee each input is really random with respect to the devices to
be used (i.e.\ there aren't any memory correlations between implementations). See the first part of
Figure 1.

For expansion, we will use the recent protocol by Miller and Shi
\cite{MS} (abbreviated as MS), which by itself gives cryptographic
security in the output  and is robust to noise. This protocol,
together with the Equivalence Lemma can take a min-entropy source
and produce unbounded expansion with only 2 untrusted devices. Following \cite{CSW,MS} we treat a device $D$ as
a black box, with which the experimenter can interact classically.
Each box $D$ will consist of $t$ spatially separated (no-signaling)
components which will play an XOR non-local game. Hence, the number
$t$ will depend on the nonlocal game to be played (e.g.\ for CHSH
$t=2$, and for GHZ $t=3$). See \cite{CHTW} for an exposition on XOR
games.

Currently, different DIQKD schemes exist that could work with our
protocol. Choosing which one to implement is a matter of taste,
since different Bell inequalities have different advantages. For
example, the protocols \cite{VV,MS} are robust against a constant
fraction of noise, while \cite{BCK2} is even safe against
no-signaling adversaries. What is common in these schemes though, is
that at least one of the parties must have access to an additional
source of randomness. Given that we would like our CDIQKD protocol to be noise tolerant, we propose to use one of \cite{VV,MS}. To our knowledge, these are the only available protocols which are secure against quantum adversaries, possessing quantum side information.

The last concepts we need to introduce are the security
parameters. The \textit{completeness error} $\eps_c$ bounds the
probability that we reject an honest implementation of the protocol,
$\mathbb{P}[\text{Reject}]\leq \eps_c$. The \textit{soundness
error} $\eps_s$ quantifies how random the output $Z$ is if we choose to
accept it. To see how, consider general output
states which are decomposed as $\Phi\circ\Gamma_{E}[\rho]=|\text{Acc}\ra \la \text{Acc}|
\otimes \sigma^{\text{Acc}}_{ZXDE} + |\text{Rej}\ra \la \text{Rej}|
\otimes \sigma^{\text{Rej}}_{ZXDE}$ , where $\Phi$ is the quantum
channel of the protocol, and $\Gamma_E$ is an arbitrary quantum channel on Eve's system. We require that there exists a
state $\xi$ such that $\xi_{ZXE}^{\text{Acc}} = U_Z\otimes \xi_{XE}$
and $ ||\sigma^{\text{Acc}}_{ZXE}- \xi_{ZXE}^{\text{Acc}}||\leq \eps_s$ .
Here, $\sigma^{\text{Acc}}_{ZXE}$ is the subnormalized output after tracing out the devices $D$, and $U_Z = \frac{1}{|Z|} \ONE$ is the uniform distribution. Most of the time though, we will just talk about the
\textit{security parameter} $\delta=\max(\eps_c,\eps_s)$, which
represents the worst error in both possible interpretations of the
word error.

The \textit{error tolerance parameter}, or noise level,
$\eta$, parametrizes how an actual implementation of an untrusted
device deviates from an honest one. That is, it is the maximum ratio of game rounds for which we observe an error (so that the observed correlations are not according to the optimal winning strategy).

\section{Key Distribution Protocol}
\begin{figure}[t]
\begin{center}
\includegraphics[scale=0.5]{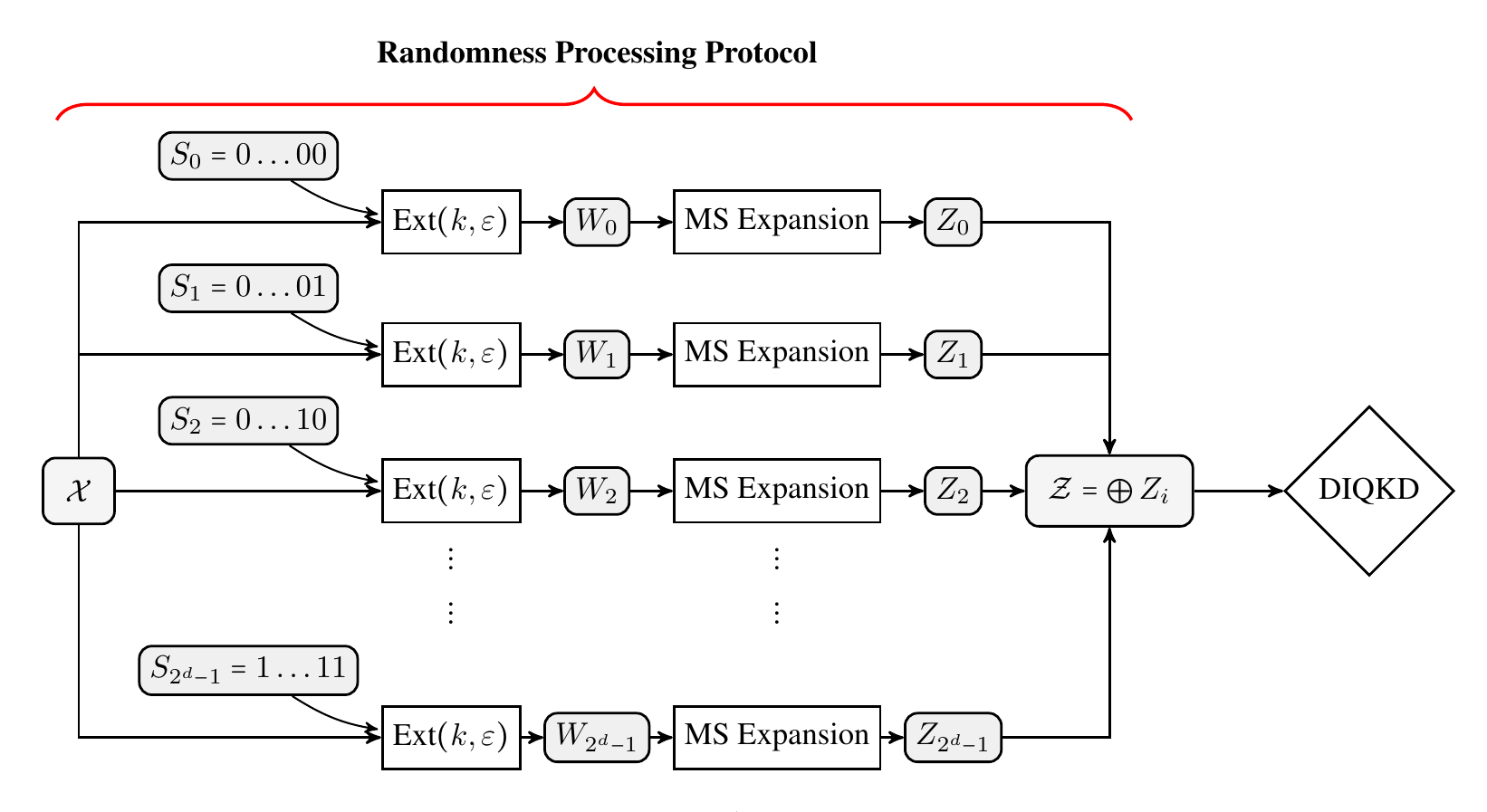}
\end{center}
\caption[]{(Color online) Schematic Representation of the Protocol. An $n$-bit
Message $\mathcal{X}$ is fed into Trevisan's Extractor with all
possible seeds $S$ of length $d$, which in turn is used to run the
Miller-Shi protocol for expansion. Finally, these outputs are summed
(modulus 2) to obtain the random seed $\mathcal{Z}$ used for the
DIQKD scheme.}
\end{figure}

For convenience, we divide our CDIQKD protocol into two parts:\ Randomness
Processing and Key Distribution. The randomness processing part
(which takes place entirely in Alice's laboratory) consists of
taking Alice's message $\mathcal{X}$ as a seed to create a string of
random numbers $\mathcal{Z}$ which will be used in the Key
Distribution Scheme (e.g.\ to choose measurement bases, which bits
to compare and test the Bell inequality on, or which hashing
function to use).

\vspace{0.5cm}\noindent\textit{Randomness Processing Protocol:}

\begin{itemize}[leftmargin=0.5cm]
\item[1.] Alice lists all possible bit strings ($S_0, S_1, \dots, S_{2^d-1}$) of length $d$.
\item[2.] Alice processes her message $\mathcal{X}$ with Trevisan's extractor, using all $2^d$ strings $S_i$ as possible seeds. Call the outputs $W_i = \text{Ext}[\mathcal{X},S_i]$.
\item[3.] Alice performs the MS unbounded randomness expansion protocol in parallel, on each $W_i$, and using different devices. The output of each expansion run is labeled $Z_i$.
\item[4.] $\mathcal{Z} = \bigoplus_{i} Z_i$
\end{itemize}
The actual size of $d=|S_i|$ and $m=|W_i|$ are specified in the next
section.


The randomness processing protocol to be used, is the composition of
the protocols proposed by \cite{CSW} and \cite{MS}, as is depicted in Figure 1. The ideal objective of the protocol is to obtain a random string $\mathcal{Z}$, independent from the input message $\mathcal{X}$, such that $|\mathcal{Z}| \gg |\mathcal{X}|$. In fact, the expansion protocol used allows us to make the output $\mathcal{Z}$ unbounded, so that Alice can be confident she will have enough random bits to feed the DIQKD protocol.

\begin{figure}[t]
\begin{center}
\includegraphics[scale=0.65]{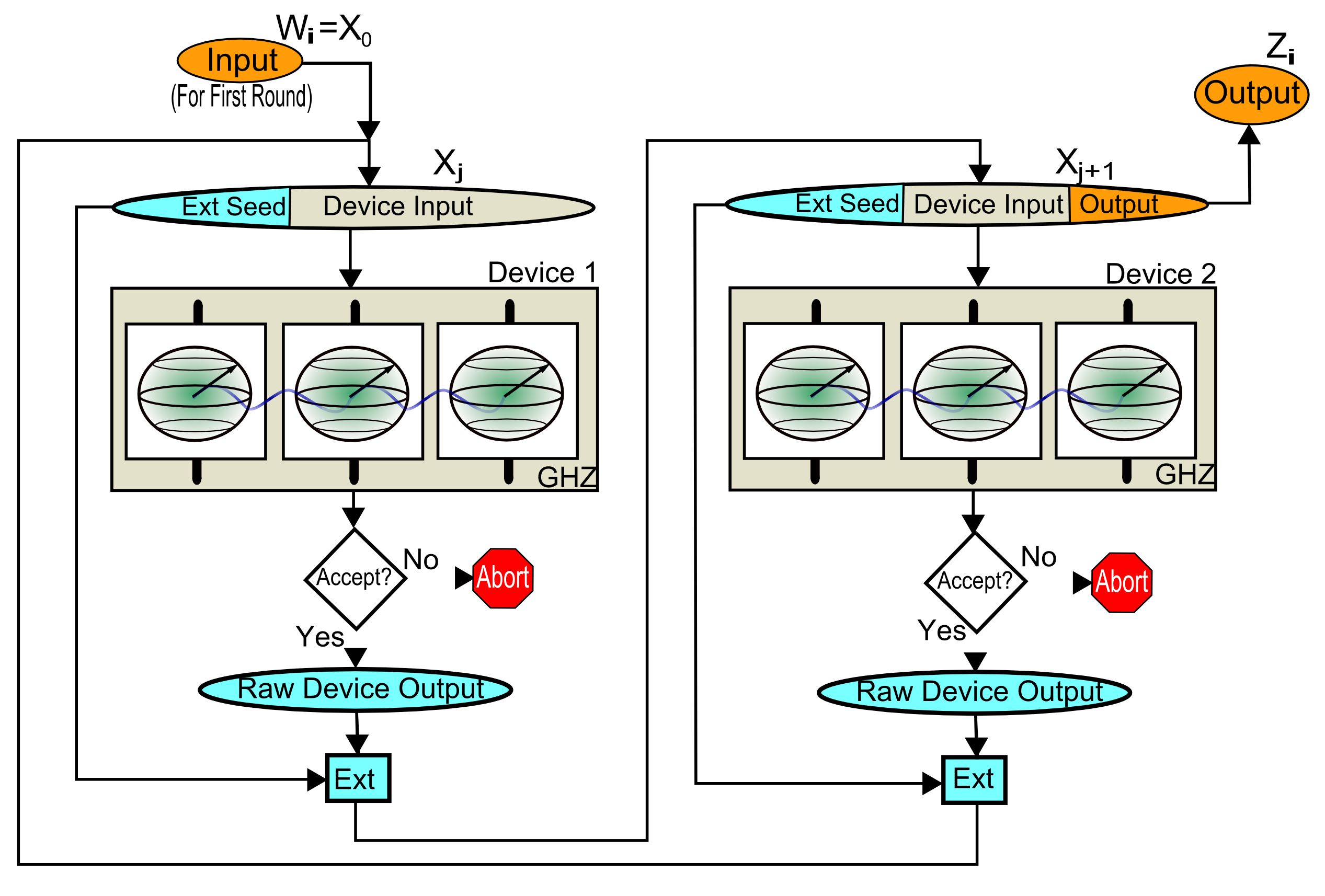}
\end{center}
\caption[]{(Color online) Representation of the MS Expansion protocol, used within the Randomness Processing Protocol. By cross feeding the outputs of the devices to each other, Alice is able to obtain the unbounded random string $Z$.
}
\end{figure}

The MS-expansion protocol uses the concatenation of two devices to achieve unbounded randomness expansion \cite{MS}. As seen in Figure 2, an input random string $X_{0}$ is fed into the first device and produces an output string $X_{1}$ which is longer and contains more min-entropy than the input. Then, string $X_{1}$ is fed into the second device, producing output $X_{2}$ which is also longer and contains more min-entropy than its corresponding input $X_1$. In this fashion, it is easily seen that alternating between the two devices, the random strings $\{X_j\}$ keep prolonging monotonically, and Alice is free to repeat this protocol as many times as possible to achieve unbounded expansion.

One may note however that in between device uses, the output must be processed through Trevisan's extractor, which provides security against quantum side information. Since the extractor requires two inputs, in reality not all of string $X_j$ is fed into a device. Rather a part of it is kept to seed the extractor, which will operate on the raw expanded output of the device. Afterwards, Alice may choose to run the expansion on the whole string $X_{j+1}$, or directly use some of the bits as an output sequence (as depicted in Figure 2).

The specific expansion protocol used to obtain the longer and more random output $X_{j+1}$, from the shorter input $X_j$ is given in Appendix C. For the moment, let's assume that Alice is running the protocol based on the GHZ non-local game, and that the size of her desired output is $N=|X_{j+1}|$. Then, Alice will feed $N$ different inputs into the components of her device which are in charge of violating the GHZ-Bell inequality. The majority of the time Alice will use a predefined input for her device's components (say $111$), and record the output of the first component (these are the so-called \textit{generating rounds}). However, in order to be sure that the components are indeed outputting random strings, Alice needs to run statistical tests on her device. For this, she will select a random subset of the $N$ inputs to actually ``play'' the GHZ game -- i.e.\ the inputs to the device components are chosen at random from the set $\{111,100,010,001\}$. The GHZ game is won if $a_1\oplus a_2 \oplus a_3 = x_1 \wedge x_2 \wedge x_3$ , where the $a_i$ are the output of the components, and the $x_i$ the corresponding inputs. If during these \textit{game rounds}, the device loses more often than allowed by the error tolerance parameter (optimized later), then Alice aborts. Otherwise, she now has a new random string $X_{j+1}$ which has more min-entropy than what she started with.

Finally, Alice will have a fully secret string $\mathcal{Z}$ with respect to Eve. If the
security of the string is high enough, this can be used to implement
the now standard protocols of \cite{VV,BCK} or even the new QKD
protocol of \cite{MS}. However, it is typically assumed that both
Alice and Bob have access to RNG's or initial randomness. Now, only
Alice has randomness available, and she must publicly broadcast to
Bob what to measure. One way for this to be secure, would be to
require that Alice and Bob were already sharing all entangled pairs
from the start. A way around this would be for Alice to wait
until Bob has received his device (i.e.\ part of the entangled
pair), and afterwards Alice would broadcast Bob's corresponding
measurement setting (see Figure 3). This eliminates the need of the vast quantum memory
of the former approach. What is needed is that there
exist quantum states and measurement settings such that each step in
the protocol would be passed by honest parties, which both
approaches possess.


\begin{figure}[t]
\begin{center}
\includegraphics[scale=1.1]{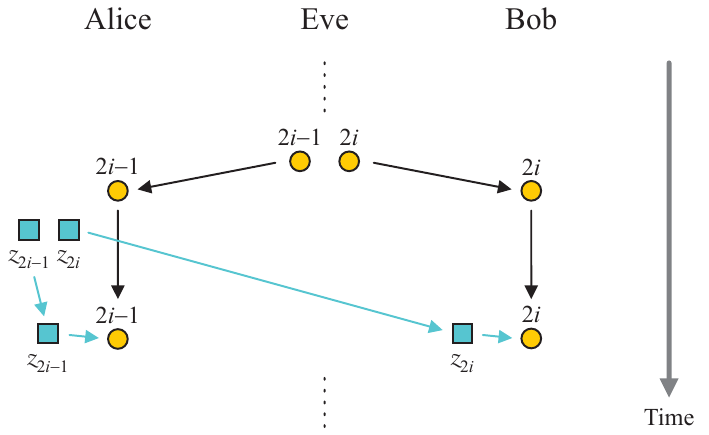}
\end{center}
\caption[]{(Color online) Space-time scheme of a QKD protocol
without initial randomness. From her secret message $\mathcal{X}$,
Alice has already established a sequence $\mathcal{Z}$ of bits $z_i$
unknown to Eve. Eve sequentially sends pairs of particles labeled
with $2i-1$ and $2i$ $(i=1,2,…)$ to Alice and Bob, respectively.
Once Bob confirms he received particle $2i$, Alice sends the bit
$z_{2i}$ to Bob, which he uses as a setting. Alice uses $z_{2i-1}$
for her own particle.}
\end{figure}

\section{Security Analysis}
In this section we analyze the security of the protocol. Our
starting point is that Alice holds a message of length $n$ that she
wants to communicate to Bob, and said message has min-entropy $k$
(conditioned on Eve and the devices). For the protocol to work, it
is part of the assumption that Alice can estimate (lower bound) the value $k$, which is
also a commonly implied assumption in other protocols such as
\cite{CR,CSW,MS}.

Of course all of the sub-protocols we are utilizing here have been
proven secure by their corresponding authors, but their composition
is a non-trivial task. Also, the point of view we take here is that
Alice has no further access to randomness, so there will be a lower
bound for the security parameters (since these are functions of $k$,
and it is finite). We also note that since this is a proof of
principle, the  requirement of exponentially many different devices
that arises from the scheme of Chung, Shi, and Wu is something we do
not intend to improve, and we rest content with having a finite
amount of devices.

The first part we analyze is Trevisan's Extractor, proven to be
secure against quantum adversaries \cite{DPVR,Tre,RRV,TSSR}. This
will create an output of size $m<k$ an $\eps_{T}$-distance from
uniform. The following Lemma gives a bound on the error and seed
length needed. For an explicit and detailed proof, see Appendix B.
\begl{Trevisan's Extractor} \nl For a message $\mathcal{X}$ with
min-entropy $k$, $0<m<k$, there exists an $m$-bit quantum proof extractor
Ext($k,\eps_T$), using a seed of length \begeq d = \left(7+k-m+\log
|\mathcal{X}| \right)^2\frac{\log (4m)}{\ln2}
\endeq
and with error \begeq \eps_{T} =
3\,m\,2^{-\frac{1}{8}(k-m)+\frac{1}{4}}.
\endeq
\endl

For analyzing the security of Miller and Shi's expansion protocol,
we must choose a nonlocal game to be played. In what follows we
shall use the GHZ game, with $t=3$. Besides having a large
quantum-classical gap and having an optimal strategy that wins with
probability 1, both \cite{MS,GA} have considered it for their
analysis. Concretely, there exist carefully optimized parameters to
implement the Miller-Shi unbounded protocol with a uniform seed,
such that the security parameter decreases exponentially with the
seed length $m$: \begeq \eps_{MS} = 2^{\frac{\alpha-m}{\beta}},
\endeq
with constants $\beta$= 31328, and $\alpha=$ 120,931. See Appendix C for further details.

It is interesting to note that while the expansion error $\eps_{MS}$
decreases exponentially with the input length $m$, the error of the
quantum proof extractor grows exponentially with the output size
$m$. Hence, there is a direct trade off, and Alice must choose $m$
accordingly to her error goals in an easy optimization problem. For
simplicity though, Alice can take e.g.\ $m=k/2$.

Finally, Chung, Shi, and Wu's main result gives the soundness and
completeness errors one obtains after having performed extraction
and expansion with each of the $2^d$ seeds and summing all of the
outputs modulo 2. The answer is a function of both the extraction
and expansion error, as well as the error tolerance $\eta$, which
comes from the Miller-Shi expansion protocol. In particular, the
security parameter $\delta$ of the whole randomness processing
protocol will be given by $\delta = \max\left( \frac{\eps_T +
\eps_{MS}}{\eta} ,\eps_{MS} + 2 \sqrt{\eps_T} +2\eta \right)$, using
a total of $6\cdot2^d$ device components \cite{CSW}. This leads us to our
first main result (proof in Appendix D).

\begt{Security of Randomness Processing}\nl If Alice performs the
Randomness Processing Protocol on her message $\mathcal{X}$ with
min-entropy $k$, the output string $\mathcal{Z}$ is
cryptographically secure. That is, the security parameter $\delta$
is exponentially small in $k$.
\endt

It is worth noting that there is some threshold value for this
protocol $k\gtrapprox$ 200,000, under which it will not work at all.
This is reminiscent of the 225,000 bits of min-entropy that are
needed to have unbounded expansion with the MS-protocol and a
security parameter of $\epsilon = 10^{-1}$ \cite{GA}. That is, in order to achieve a fixed security parameter target for randomness expansion, the amount of input min-entropy must be above some threshold. In any case,
we imagine $k$ to be large enough so that the security parameter is
sufficiently small.



Now that Alice has the random string $\mathcal{Z}$, she is ready to
apply, together with Bob, the DIQKD protocol of either \cite{VV} or \cite{MS}, which have their respective errors $\eps_c, \eps_s$. For a moment let us assume that $\mathcal{Z}$ is a perfectly random
string, then the Equivalence Lemma of \cite{CSW} would guarantee
that the completeness and soundness errors of the DIQKD protocol
would remain the same even if Eve learned most of $\mathcal{Z}$
later on (making this semantically secure). However, $\mathcal{Z}$ has security parameter $\delta$,
exponentially small in $k$, and this will add to the errors of the
protocol (which could be understood as a consequence of the composability of the protocols \cite{MR}). Note that the string $\mathcal{Z}$ is indeed random to the devices
in the DIQKD protocol, since the initial message had min-entropy $k
= H_{\text{min}}(\mathcal{X}|ED)$ conditioned on both the randomness
processing devices $D$, and Eve (who is the one who potentially will
create the DIQKD devices). We formalize this in our second main theorem, which is proven in Appendix E.

\begt{Security of CDIQKD}\nl Let there be a DIQKD protocol which
requires a perfect Random Number Generator and which has completeness and
soundness errors $(\eps_c,\eps_s)$. Then, Alice can perform the
Randomness Processing Protocol on her secret message $\mathcal{X}$
with min-entropy $k$, to produce a secure random output
$\mathcal{Z}$ and perform CDIQKD with errors
$(\eps_c+\delta,\eps_s+\delta)$, where $\delta = 2^{
-\Omega\left(k\right)}$.
\endt


\section{Discussion}
We have shown that even in the absence of randomness generators,
Alice can securely perform DIQKD. This is indeed a remarkable fact,
since it is commonly assumed that without initial randomness no
security can be achieved. In this article, we have made a proof of
principle based on the assumptions given. Note however, that our protocol still required the use of a classical authenticated channel which traditionally is established using a shared secret key between the honest parties. At first sight this seems to call into question the result of this paper. However remark that the authenticated channel does not have to be established each time the parties wish to send a message to each other. As stated in Assumption 2, we consider the authenticated channel to be a black box resource, that the parties could have established a long time in the past. A shared arbitrarily weak key suffices for this task, as shown in \cite{RW03}. Traditional DIQKD relies on a further assumption, namely that the parties hold private secure random number generators, which they use to obtain inputs for the protocol. The security of the output randomness of these RNGs could be subject to question especially if these were prepared by an external adversary. The issue this paper addresses is therefore the removal of this crucial assumption in a fairly general framework for DIQKD. Finally, a secret key shared by the parties could replace the message in the presented protocol if it is of sufficiently high min-entropy. 

 We leave further generalizations and optimizations for future work. For example, we conjecture that our scheme can be simplified to use a significantly smaller number of devices and that it can be generalized to be secure against no-signaling adversaries also, leading to drop the validity of quantum mechanics as an assumption.

\textit{Acknowledgements.}  E.A.\ would like to thank R.\ Renner and K.\ Horodecki for
helpful discussions, and M.\ Farkas for reviewing the manuscript. This work was supported by NCN grant
2013/08/M/ST2/00626, the IDSMM program of the University of Gdansk,
and the National Quantum Information Centre in Gdansk. J.K.\
acknowledges support from the European Union Integrated Project
Simulators and Interfaces with Quantum Systems. R.R.\ acknowledges support from the ERC AdG grant QOLAPS, and the Foundation for Polish Science TEAM project co-financed by the EU European Regional Development Fund.

\onecolumn
\newpage
\appendix
\begin{center}
\Large{\textbf{Appendix}}
\end{center}

\section{Definitions and Notation}
In this section, we formalize some important definitions, which were
just mentioned conceptually in the main text. Throughout this whole
article, as is common in information science, $\log(x)=\log_2(x)$.

\begd{Conditional Min-Entropy}\nl Let
$\rho_{AB}\in\mathcal{D}(\mathcal{H}_A\otimes\mathcal{H}_B)$, the
min-entropy  of $A$ conditioned on $B$ is: \begeq
H_{\text{min}}(A|B)_{\rho_{AB}} = \max\{\lambda\in\RR : \exists \SP
\sigma_B \in \mathcal{D}(\mathcal{H}_B) \SP \text{s.t.} \SP
\rho_{AB} \leq 2^{-\lambda} \ONE_A \otimes \sigma_B \}
\endeq
\endd

Here, $\mathcal{D}(\mathcal{H})$ represents the set of density
matrices in Hilbert space $\mathcal{H}$. For the completeness and
soundness errors, we use the definitions given by \cite{CSW}, since
our security parameter is based on the maximum of these quantities.
Before that, we must specify what is meant by a physical system.

\begd{Physical System \cite{CSW}}\nl A physical system $\mathcal{S}$
is defined on an arbitrarily large, but finite Hilbert space
$X\otimes D \otimes E$, with a classical source $\mathcal{X}$ of
length $n$ with $k$ min-entropy, $t$ untrusted devices
$D=(D_1,\dots,D_t)$, and an adversary $E$. To each device $D_i$
there corresponds a quantum interactive algorithm $A_{D_i}$ that
applies on $D_i$ which outputs at most $m$ bits.
\endd

Usually this is also called an $(n,k,t,m)$-Physical Source, where in
our scenario the min-entropy $k$ the message has is conditioned on
both $E$ and $D$. Hence, a physical system $\mathcal{S}$ is
specified by a state $\rho_{XDE}$ and the algorithms the devices
will follow $\{A_{D_i}\}$, but the latter are usually irrelevant for
the security analysis.

Any randomness processing protocol (e.g.\ for amplification or
expansion) can be viewed as a quantum channel
$\Phi:\mathcal{D}(X\otimes D) \rightarrow \mathcal{D}(O \otimes Z
\otimes X \otimes D)$, also called \textit{Physical Randomness
Extractors} (since they act on physical systems). The new Hilbert
spaces $O \otimes Z$ are for a decision bit $o$ which will tell us
to accept or reject the implementation of the protocol (if e.g.\ the
Bell test was not passed with confidence), and the new output random
string $\mathcal{Z}$. If the physical randomness extractors require
perfectly random inputs, i.e.\ they are designed to work on
$(n,n,t,m)$-physical systems, they are called \textit{Seeded
Physical Randomness Extractors}.

\begd{Completeness Error \cite{CSW}}\nl There exist honest devices
$D=(D_1,\dots,D_s)$ with internal state $\sigma_D$ and algorithms
$\{A_{D_i}\}$, with each device outputting at most $m$ bits such
that for any $(n,k,s,m)$-physical system $\mathcal{S}$ satisfying
$\tr_{XE}[\rho]=\rho_D=\sigma_D$, we have \begeq
\mathbb{P}[\text{Acc}(\rho)]\geq 1 - \eps_c
\endeq
Where Acc$(\rho)$ denotes the event that the protocol accepts on the input state of the device and source supplied to the physical randomness extractor, when applied to $\mathcal{S}$, (i.e.\ $o=\text{Acc}$).
\endd

In other words, this tells us that if we are using honest devices,
we will accept the protocol with high probability. The soundness
error, in turn tells us how close we are to a truly random output
(i.e.\ a uniform distribution), conditioned on accepting the
protocol.

\begd{Soundness Error \cite{CSW}}\nl Suppose the physical system
$\mathcal{S}$ is equipped with a decision bit $O$, then the projection of the output
$\Phi[\rho]$ to the Acceptance subspace is at most an $\eps_s$ distance away from a state of
the form $U_Z\otimes \xi_{XE}$ conditioned on accepting, where
$\xi_{XE}$ is some classical quantum state. General output
states are decomposed as $\Phi\circ\Gamma_{E}[\rho]=|\text{Acc}\ra \la \text{Acc}|
\otimes \sigma^{\text{Acc}}_{ZXDE} + |\text{Rej}\ra \la \text{Rej}|
\otimes \sigma^{\text{Rej}}_{ZXDE}$, where $\Gamma_E$ is an arbitrary quantum channel on Eve's system. We require that there exists a
state $\xi$ such that $\xi_{ZXE}^{\text{Acc}} = U_Z\otimes \xi_{XE}$
and \begeq ||\sigma^{\text{Acc}}_{ZXE}- \xi_{ZXE}^{\text{Acc}}||\leq \eps_s
\endeq
Where $\sigma^{\text{Acc}}_{ZXE}$ is the subnormalized output after tracing out device $D$, and $U_Z = \frac{1}{|Z|} \ONE$ is the uniform distribution.
\endd

An important result which we will need for our analysis, which has
to do with physical randomness extractors, is the following lemma:

\begl{Equivalence Lemma \cite{CSW}}\nl Let $\Phi$ be a seeded
physical randomness extractor, with seeds $X$ which are perfectly
random to both Eve and the Devices (i.e.\ $H_\text{min}(X|DE)=n$),
have parameters $(\eps_s,\eps_c,\eta)$. Then the same physical
randomness extractor $\Phi$, when applied to an input which is
perfectly random to just the devices (i.e.\ $H_\text{min}(X|D)=n$)
will have the same parameters $(\eps_s,\eps_c,\eta)$.
\endl

The moral being, that the crucial thing is that the input is random
to the devices used.

\section{Quantum Strong Extractors}

In this section, we analyze the security of Trevisan's extractor
from reference \cite{DPVR}, to prove Lemma 4.1. We begin with a
formal definition of a quantum-strong extractor. \begd{Quantum Proof
Strong Extractor}\nl Ext:$\{0,1\}^n \times \{0,1\}^d \rightarrow
\{0,1\}^m$ , is an $m$-bit quantum proof $(k,\eps)$-strong
extractor, if for all states $\rho_{XE}$ classical on $X$ with
$H_{\text{min}}\geq k$, and a uniform seed $Y$ of length $d$, we
have: \begeq \frac{1}{2}|| \rho_{\text{Ext}(X,Y)YE} -
U_m\otimes\rho_y\otimes \rho_E|| \leq \eps
\endeq
With $||\cdot||$ the trace-norm, and $U_m$ the totally mixed state in $\CC^{2^m}$.
\endd
The classical version of this definition ignores the quantum state
$E$ and uses the variational distance in equation (B.1). Explicitly,
a $(k,\eps)$-\textit{strong extractor} satisfies
$\frac{1}{2}||Ext(X,Y)\circ Y - U_m \circ Y || \leq \eps$. The main
theorem of \cite{DPVR} relates the security of 1 bit
$(k,\eps)$-strong extractors to $m$-bit extractors which are quantum
proof. This is done via means of \textit{weak $(t,r)$-designs},
which are just families of partioning sets -- otherwise irrelevant
here.

\begt{Trevisan's Extractor is Quantum Proof (Theorem 4.6 of
\cite{DPVR})}\nl Let $C:\{0,1\}^n \times \{0,1\}^t \rightarrow
\{0,1\}$ be a $(k,\eps)$-strong extractor with uniform seed and
$S_1,...,S_m \subset [d]$, a weak $(t,r)$-design. Then $\exists$ an
extractor $\text{Ext}_C:\{0,1\}^n \times \{0,1\}^d \rightarrow
\{0,1\}^m$, which is a $(k+rm+\log(1/\eps),3m\sqrt{\eps})$-quantum
proof strong extractor.
\endt

The existence of such weak designs is given by \cite{RRV}.

\begl{Existence of weak $(t,1)$-designs (Lemma 17 of \cite{RRV})}\nl
$\forall t, m\in \NN, \exists$ weak $(t,1)$-design $S_1,...,S_m
\subset [d]$ such that $d=t\lceil{\frac{t}{\ln
2}}\rceil\lceil{\log(4m)}\rceil$. Furthermore such a weak-design can
be found in Poly$(m,d)$ time and Poly$(m)$ space.
\endl

For the 1-bit extractor $C$, we will use \textit{list-decodable
codes} -again, for our purposes all we require is their existence
and that they can be found efficiently. This was implicitly proven
by \cite{Tre,RRV}, and explicitly stated in \cite{DPVR} Theorem C.3.

\begl{List Decodable Codes are 1-bit Extractors (Theorem C.3 of
\cite{DPVR})}\nl Let $C:\{0,1\}^n \rightarrow \{0,1\}^{\bar{n}}$ be
an $(\eps,L)$-list decodable code. Then $\exists$
$\text{Ext}_C:\{0,1\}^n\times[\bar{n}]\rightarrow\{0,1\}$, which is
a $(\log L + \log\left(\frac{1}{2\eps}\right) , 2\eps)$-strong
extractor, created from code $C$.
\endl

Finally, we need an existence theorem for list decodable codes.

\begl{Existence of List Decodable Codes (Lemma C.2 of \cite{DPVR},
Theorem 24 of \cite{GHSZ})}\nl $\forall n\in \NN$, and $\eps>0$,
$\exists$ a code $ C_{n,\eps}:\{0,1\}^n\rightarrow
\{0,1\}^{\bar{n}}$ which is $(\eps,1/\eps^2)$ list decodable.
Furthermore $\bar{n}$ can be assumed to be a power of 2, and
satisfies the bound $\bar{n}\leq 32 n/\eps^4$. The code $C_{n,\eps}$
can be evaluated in Poly$(n,1/\eps)$ time.
\endl

With all of this in mind, we are ready to prove Lemma 4.1. This is
an analogous result to Corollary 5.3 of \cite{DPVR}, and \cite{GA}
has also made a similar analysis.

\begl{(Lemma 4.1 from main text)}\nl For a message $\mathcal{X}$
with $k$ min-entropy, $m<k$, there exists an $m$-bit quantum proof
extractor Ext($k,\eps_T$), using a seed of length \begeq d =
\left(7+k-m+\log |\mathcal{X}| \right)^2\frac{\log (4m)}{\ln2}
\endeq
and with error \begeq \eps_{T} =
3\,m\,2^{-\frac{1}{8}(k-m)+\frac{1}{4}}.
\endeq
\endl

\begin{proof}Lemma 4.1\nl
To facilitate the proof of this lemma, which involves many different
concepts and parameters, we refer to Figure 4. Notice
that the notation is slightly different from the statement of the
lemmas, to make it more consistent throughout the proof.

\begin{figure}[t]
\begin{center}
\includegraphics[scale=0.85]{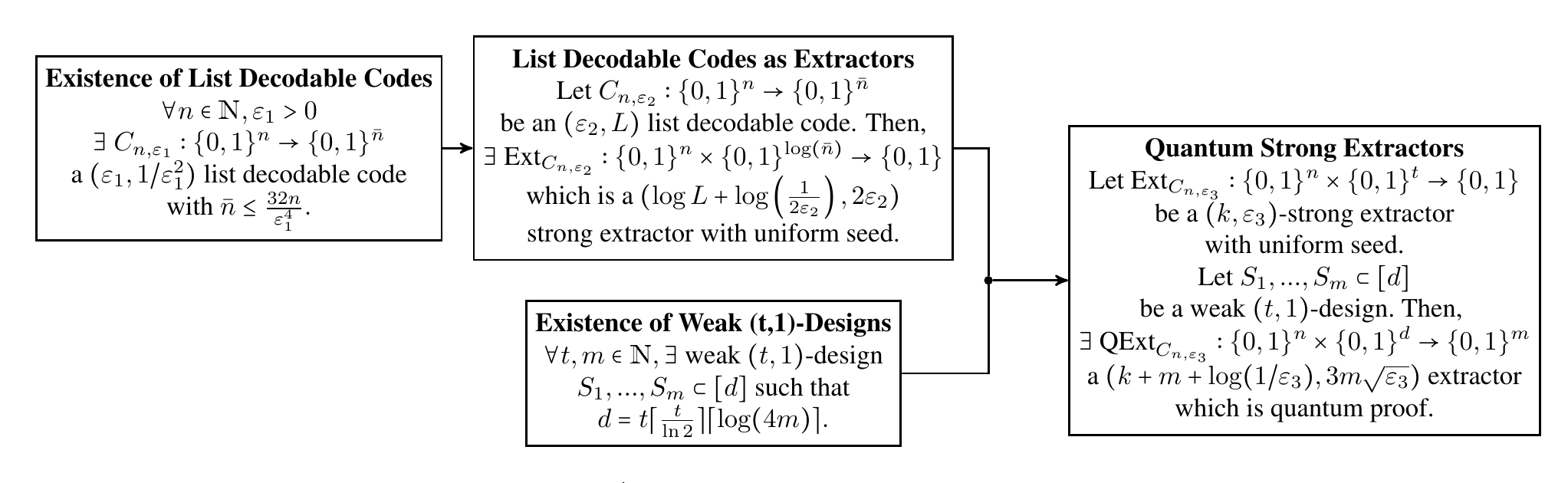}
\end{center}
\caption[]{Schematic diagram for the proof of Lemma 4.1.}
\end{figure}

We take $\bar{n}$ to be a power of 2, $\bar{n}=2^t$. Next,  we
create a $(\delta,1/\delta^2)$ List-Decodable Code from Lemma B.3,
so that we are guaranteed the existence of a 1-bit
$(3\log\left(\frac{1}{2\delta}\right)+2,2\delta)$-strong extractor
$C:\{0,1\}^n\times\{0,1\}^t\rightarrow\{0,1\}$, with the help of
Lemma B.2. We consider the worst case (saturated) bound on
$\bar{n}$: \begeq \bar{n} = \frac{32n}{\delta^4} \rightarrow t =
\log \left( \frac{2^5 n}{\delta^4} \right)
\endeq

Equipped with this 1-bit extractor, we shall now use Theorem B.1 to
create an $m$-bit extractor which is quantum proof. Direct
application of the Theorem yields a $(4
\log\left(\frac{1}{2\delta}\right) + m + 2,
3m\sqrt{2\delta})$-quantum proof extractor. We want the final error
of the extractor to be $\eps$, hence we take
$\delta=\eps^2/(2\cdot9m^2)$, to get a quantum proof
$(8\log\left(\frac{m}{\eps}\right) + m + 2 + 8\log3 , \eps)$-strong
extractor. Now, in order for this extractor to work, we need the
min-entropy of the input message to satisfy $k\geq
8\log\left(\frac{m}{\eps}\right) + m + 2 + 8\log3$. Manipulating
this inequality gives us the minimal error the output of the
extractor can have. \begeq \eps \geq 3\,m\,2^{\frac{2+m-k}{8}}
\endeq
Finally, the $t$ from Equation (B.2) is the same appearing in Lemma
B.1, related to the $(t,1)$ designs. Since we are bounding the
number of devices (and hence seed length), we will ignore the
ceiling operators from Lemma B.1, which in the limit of large $t$
and $m$ will be negligible. Hence, the seed length for Trevisan's
extractor will be $d = \log^2\left( 2^9 3^8\frac{ n m^8 }{\eps^8}
\right)\frac{\log(4m)}{\ln2}$ (having substituted in the value for
$\delta$). If now, we take the lowest bound from Equation B.5 for
the error $\eps$ we obtain \begeq d = \left(7+k-m+\log n
\right)^2\frac{\log (4m)}{\ln2}
\endeq
It is interesting to note that the error $\eps$ only depends on $m$
and $k$, having a direct trade off between the available min-entropy
and how large of an output we desire. Meanwhile, $d$ depends on all
parameters but the term $k-m$ has opposite sign, showing
qualitatively that the error and seed length are inversely related.
\end{proof}

\section{Randomness Expansion}

In this section, we explicitly analyze the protocol that we are
using for expansion, namely the one given by Miller and Shi
\cite{MS}. In particular, we choose this protocol since it provides
cryptographic security, i.e.\ the error parameters are exponentially
small and are negligible in the running time of the protocol. It
also tolerates a constant level of noise, where e.g.\ it was shown
that any device which wins the GHZ game with probability at least
0.985 will achieve exponential randomness expansion with probability
approaching unity. Finally, and very important for us, with the
Equivalence Lemma (as given by \cite{CSW}) this is able to produce
unbounded expansion using only two devices -- by realizing that the
expansion protocol is indeed a physical randomness extractor.

In what follows, for simplicity, we will restrict the protocol to
playing the GHZ game where the optimum quantum strategy wins with
probability 1, and will refer the readers to \cite{MS,GA} for the
generic version.

The unbounded protocol, is just a concatenation of their one-shot
protocol, so we provide the latter here. For that, we need to define
the variables needed: $N\in \NN$, is the\textit{ output length},
$\eta\in(0,\frac{1}{2})$ the \textit{error tolerance}, denoting how
much of a statistical error the components are allowed to make
relative to the optimal winning strategy's expectation, and
$q\in(0,1)$ the \textit{test probability} which denotes the chance a
given round will be a game (Bell) round. The protocol is then:

\begin{enumerate}
\item A bit $g$ is chosen according to the distribution $(1-q,q)$.
\item If $g=1$ ("game round"), then an input string from $\{111, 001, 010, 100\}$ is chosen at random to play the GHZ game. If the GHZ game is won then output $0$, else output $1$ and record "Failure" $F$.
\item If $g=0$ ("generating round"), the string $111$ is used as input on the device $D=(D_1,D_2,D_3)$. Record the output of the first component $D_1$.
\item Repeat steps 1--3, $(N-1)$ more times.
\item If the total number of failures $F$ exceeds $\eta q N$, the protocol \textit{Aborts}. Otherwise, the protocol \textit{Succeeds}, and the output $N$-bit sequence is recorded.
\end{enumerate}

In general, the one-shot protocol as given above can (for the right
choice of parameters $\eta, q , N$) provide an output which is
$\eps$-close to having $(1-\delta)N$ min-entropy for any choice of
$\delta$, and $\eps$ exponentially small as a function of $N$. Gross
and Aaronson have optimized over the parameters $(\eta , q , N)$ and
given a bound on the initial seed length needed to get unbounded
expansion \cite{GA}. In particular, they display a linear dependence
on $\log(1/\eps)$, giving the actual slope to be $\beta =$ 31328.
Then, they state that the upper bound on seed length needed to get
security of $\eps=10^{-1}$ is 225,000. From this, simple
substitution gives $\alpha \leq 120931$, and hence Lemma 4.2. We
note that they also give a bound of 715,000 bits needed to achieve
$\eps=10^{-6}$, but this gives a lower value of $\alpha$ (= 90,584),
so we conservatively kept the upper bound. For asymptotic
statements, these constants are irrelevant so long they remain
positive.

\section{Security of Randomness Processing}

In this section, we follow the analysis of \cite{CSW}, to prove the
security of our randomness processing protocol, as given in Section
3 of the main text.

Hence for our analysis, the following theorem is crucial.

\begt{Chung-Shi-Wu Theorem \cite{CSW}}\nl Let $0<\eta<1$ be the
error tolerance parameter. Let $X$ be an $n$-bit string with $k$
min-entropy. Let Ext$(k,\eps_T):\{0,1\}^n \times \{0,1\}^d
\rightarrow \{0,1\}^m$, be an $m$-bit quantum proof extractor, with
seed length $d$. Let there be a protocol $\Phi$ (also called
physical randomness extractor), which takes a perfectly random seed
of length $m$ to produce an output random string $z$,  together with
a decision bit $\mathcal{o}$, with completeness error $\eps_c$ and
soundness error $\eps_s$. If for every $S_i\in\{0,1\}^d$, we perform
$\Phi[Ext[X,S_i]]=Z_i$, ($Ext[X,S_i]$ being Trevisan's extractor
applied on string $X$ using seed $S_i$), then the protocol producing
the output string $\mathcal{Z}= \bigoplus Z_i$ has: \nl Completeness
Error $\frac{\eps_{c}+\eps_T}{\eta}$\nl Soundness Error
$\eps_{s}+2\sqrt{\eps_T}+2\eta$\nl provided less than an
$\eta$-fraction of protocol $\Phi$ applications were rejected.
\endt

This is the exact form of the randomness processing protocol that we
have given in the main text (Figure 1), where $\Phi$ will be Miller
and Shi's unbounded expansion protocol. We will take the
MS-expansion security parameter $\eps_s = \eps_c = \eps_{MS}$ given
by \cite{GA}. Here, we are still left with our errors as functions
of $m$ , $k$ , and now (from the previous theorem) $\eta$. To have a
bound on the security parameter, we will find explicit functions for
$m$ and $\eta$, depending only on Alice's min-entropy $k$. We thus
have all the ingredients to prove Theorem 4.1. We note however that
since this article focuses on a proof of principle, the following
proof is done such that it is clear to follow at the expense of not
choosing the most optimum coefficient for the exponential decay in
the security parameter.

\begin{theorem}{Security of Randomness Processing (Theorem 4.1)}\nl
If Alice performs the Randomness Processing Protocol on her message
$\mathcal{X}$ with min-entropy $k$, the output string $\mathcal{Z}$
is cryptographically secure. That is, $\exists$ a constant
$\gamma>0$ such that the security parameter is $\delta =
O\left(2^{-\gamma k}\right)$.
\end{theorem}

\begin{proof}{Theorem 4.1}\nl
We begin by writing the error parameters that arise from the
Trevisan Extractor as $\eps_T=c_1 m\, 2^{-c_2(k-m)}$, and the one
from the Miller and Shi expansion protocol as $\eps_{MS}=c_3\,
2^{-c_4 m}$, for some constants $c_1,c_2,c_3,c_4$. Where $k$ is the
min-entropy of the message $\mathcal{X}$, and $m$ is both the output
length of Trevisan's Extractor, and the input size of the expansion
protocol. For simplicity, we shall take $m=k/2$, which yields as
errors:
\begin{align*}
\eps_T &=\frac{c_1}{2} k\, 2^{-\frac{c_2}{2}k}  \\
\eps_{MS}& =c_3\, 2^{-\frac{c_4}{2} k}
\end{align*}

From the Chung-Shi-Wu theorem D.1, we have that the security
parameter $\delta = \max \left(\frac{\eps_{MS}+\eps_T}{\eta}
,\eps_{MS}+2\sqrt{\eps_T}+2\eta \right)$. So we will take $\eta =
2^{-\alpha k}$, with a suitably chosen $\alpha$. To make the
security parameter as small as possible, we must choose $\alpha$
large enough so it does not dominate the soundness error but we see
that this will bring a trade-off with the completeness error. In
fact, from the completeness error, we have:
\[
\frac{\eps_{MS}+\eps_T}{\eta} = \frac{c_1}{2} k\,
2^{-\left(\frac{c_2}{2} - \alpha\right)k} + c_3\,
2^{-\left(\frac{c_4}{2} -\alpha\right) k}
\]
This requires that $2\alpha < \min(c_2,c_4)$. From the
soundness error we have:
\[
\eps_{MS}+2\sqrt{\eps_T}+2\eta = c_3\, 2^{-\frac{c_4}{2} k} +
\sqrt{2c_1 k}\, 2^{-\frac{c_2}{4}k} + 2\cdot2^{-\alpha k}
\]
From here, for the asymptotic statement, we see that we need a
choice of $\alpha$ such that the expression \nl
$\min\left(\frac{c_2}{2}-\alpha,\frac{c_4}{2}-\alpha,\frac{c_2}{4},\frac{c_4}{2},\alpha\right)$
is as big as possible, since those are the coefficients of the
exponential decay. Using our actual values for the constants $c_2 =
1/8$ (from Lemma 4.1), and $c_4=1/31328$ from \cite{GA}, we see that
the best is to take $\alpha=c_4/4=1/125312$. This completes the
proof of the theorem, with the security parameter as $\delta =
O\left(2^{-\gamma k}\right)$ with $\gamma \geq 1/125312$ (since our
choice of $m<k$ was the simplest).
\end{proof}

In fact, using the same values for the constants $c_2,c_4$, if we
instead take $m=\frac{3916}{3917}k$, and $\eta=2^{-k/62672}$, we can
get a better $\tilde{\gamma}=1/62672$, which is almost a factor 2
better than the exponent given in the theorem. Note further that for
asymptotic statements, for any $\eps,\alpha>0$, we have
$\text{poly}(x)e^{-\alpha x} = O(e^{-(\alpha -\eps)x})$, so that we
essentially ignore the prefactors $k$ and $\sqrt{k}$ which appear in
the proof.

\section{Security of CDIQKD}

In this last section, we prove Theorem 4.2. \begt{Security of CDIQKD
(Theorem 4.2)}\nl Let there be a DIQKD protocol which requires a
perfect Random Number Generator, and has completeness and soundness
errors $(\eps_c,\eps_s)$. Then, Alice can perform the Randomness
Processing Protocol on her secret message $\mathcal{X}$ with
min-entropy $k$, to produce a secure random output $\mathcal{Z}$ and
perform CDIQKD with errors $(\eps_c+\delta,\eps_s+\delta)$, where
$\delta = 2^{-\Omega \left( k \right)}$.
\endt
\begin{proof}{Theorem 4.2}\nl
We need to check the security of the DIQKD protocol, given that
$\mathcal{Z}$ is not perfectly random, but rather has an
exponentially small error, $\delta=2^{-\gamma k}$, for constant $\gamma>0$. Let $\rho_{ZXDE}$ be the output state of
the randomness processing protocol (conditioned on accepting), then
the soundness error just means that $||\rho_{ZXDE} - U_Z \otimes
\rho_{XDE}||\leq \delta$, where $D$ refers to the devices of the DIQKD protocol.\nl \textit{Completeness:} The DIQKD
completeness error $\eps_c$ is calculated expecting perfect
randomness in the protocol. Hence
$\mathbb{P}[\text{Rej}(U_Z\otimes\rho_{XDE})]\leq \eps_c$. Here
Rej[$\rho$] denotes the event that the protocol rejects upon input
$\rho$. This immediately implies that
$\mathbb{P}[\text{Rej}(\rho_{ZXDE})]\leq \eps_c + \delta$, since the
trace norm operationally corresponds to the distinguishability
of states.\nl
\textit{Soundness:} Let $\Lambda$ be the quantum channel of the
DIQKD protocol which produces the shared key $\mathcal{Y}$ between
Alice and Bob, we write $\Lambda^{\text{Acc}}$ to denote the action of the quantum channel upon acceptance. Let $\Lambda^{\text{Acc}}[U_Z\otimes\rho_{XDE}]=\sigma^{\text{Acc}}_{YZXDE}$, then the soundness error is given by
$||\sigma^{\text{Acc}}_{YZXE} - U_Y \otimes
\sigma_{ZXE}^{\text{Acc}}||\leq \eps_s$. From the contractivity of
the trace norm under quantum channels, we have
$||\Lambda^{\text{Acc}}[\rho_{ZXDE}] - \Lambda^{\text{Acc}}[U_Z \otimes \rho_{XDE}]||\leq
||\rho_{ZXDE} - U_Z \otimes \rho_{XDE}||\leq \delta$. And by the
triangle inequality, we have the new soundness error
$||\psi^{\text{Acc}}_{YZXE}- U_Y \otimes \sigma_{ZXE}^{\text{Acc}}||\leq
\eps_s + \delta$, with $\Lambda^{\text{Acc}}[\rho_{ZXDE}]=\psi^{\text{Acc}}_{YZXDE}$.
\end{proof}

\newpage

\bibliographystyle{ieeetr}
\bibliography{msprefs}

\end{document}